\documentclass{article}
\usepackage{graphicx}

\usepackage{todonotes}

\usepackage{amsmath}
\usepackage{amssymb}
\usepackage{amsfonts}
\usepackage{amsthm}
\usepackage{yfonts}

\usepackage[portrait, left=1in,right=1in,top=1in,bottom=1in] {geometry}

\usepackage{algorithm}
\usepackage[noend]{algpseudocode}

\theoremstyle{plain}
\newtheorem{theorem}{Theorem}
\newtheorem{corollary}[theorem]{Corollary}
\newtheorem{lemma}[theorem]{Lemma}
\newtheorem{proposition}[theorem]{Proposition}

\theoremstyle{definition}
\newtheorem{definition}[theorem]{Definition}

\newtheorem{remark}[theorem]{Remark}

\usetikzlibrary{patterns, snakes, calc}

\newcommand{\rr}{\mathbb{R}}
\newcommand{\nn}{\mathbb{N}}

\theoremstyle{remark}
\author{Greg Bodwin and Lily Wang\\University of Michigan EECS\\ \texttt{\{bodwin, lilyxy\}@umich.edu}}
\date{}

\begin{document}
\title{Improved Shortest Path Restoration Lemmas for Multiple Edge Failures: Trade-offs Between Fault-tolerance and Subpaths\thanks{This work was supported by NSF:AF 2153680}}
\maketitle
\noindent

\begin{abstract}
The \emph{restoration lemma} is a classic result by Afek, Bremler-Barr, Kaplan, Cohen, and Merritt [PODC '01], which describes how the structure of shortest paths in a graph can change when some edges in the graph fail.
Their work shows that, after one edge failure, any replacement shortest path avoiding this failing edge can be partitioned into two pre-failure shortest paths.
More generally, this implies an \emph{additive} tradeoff between fault tolerance and subpath count: for any $f, k$, we can partition any $f$-edge-failure replacement shortest path into $k+1$ subpaths which are each an $(f-k)$-edge-failure replacement shortest path.
This generalized version of the result has found applications in routing, graph algorithms, fault tolerant network design, and more.

Our main result improves this to a \emph{multiplicative} tradeoff between fault tolerance and subpath count.
We show that for all $f, k$, any $f$-edge-failure replacement path can be partitioned into $O(k)$ subpaths that are each an $(f/k)$-edge-failure replacement path.
We also show an asymptotically matching lower bound.
In particular, our results imply that the original restoration lemma is exactly tight in the case $k=1$, but can be significantly improved for larger $k$.
We also show an extension of this result to weighted input graphs, and we give efficient algorithms that compute path decompositions satisfying our improved restoration lemmas.
\end{abstract}

\section{Introduction}

Suppose we want to route information, traffic, goods, or anything else along shortest paths in a distributed network.
In practice, network edges can be prone to \emph{failures}, in which a link is temporarily unusable as it awaits repair.
It is therefore desirable for a system to be able to adapt to these failures, efficiently rerouting paths on the fly into new replacement paths that avoid the currently-failing edges.
An algorithm that repairs a shortest path routing table following one or more edge failures is called a \emph{restoration algorithm} \cite{restoration}; the design of effective restoration algorithms forms a large and active body of work in both theory and practice, see e.g., \cite{liu2020survivability, ZGKLXZ21, WGYJ+15, restoration, tiebreaking} and references within.
The focus of this paper will specifically be on restoration algorithms that recover \emph{exact} shortest paths in the post-failure graph.

An ideal restoration algorithm will avoid recomputing shortest paths from scratch after each new failure event, instead leveraging its knowledge of the pre-failure shortest paths to speed up computation.
Therefore, when designing restoration algorithms, it is often helpful to understand exactly how shortest paths in a graph can evolve following edge failures.
A \emph{restoration lemma} is the general name for a structural result relating the form of pre-failure shortest paths to post-failure shortest paths in a graph, named for their applications in restoration algorithms.

The original restoration lemma was made explicit in a classic paper by Afek, Bremler-Barr, Kaplan, Cohen, and Meritt \cite{restoration}, and it was implicit in work before that, such as \cite{KIM82}.
All graphs in this discussion are undirected and unweighted, until otherwise indicated.

\begin{definition} [Replacement Paths]
A path $\pi$ in a graph $G = (V, E)$ is an \emph{$f$-fault replacement path} if there exists a set of edges $F \subseteq E, |F| \le f$ such that $\pi$ is a shortest path in the graph $G \setminus F$.
\end{definition}

\begin{theorem} [Original Restoration Lemma \cite{restoration}] \label{thm:baserestoration}
In any graph $G$, every $f$-fault replacement path can be partitioned into $f+1$ subpaths that are each a shortest path in $G$.
\end{theorem}

\begin{figure}[h]
\begin{center}
\includegraphics[scale = 0.7]{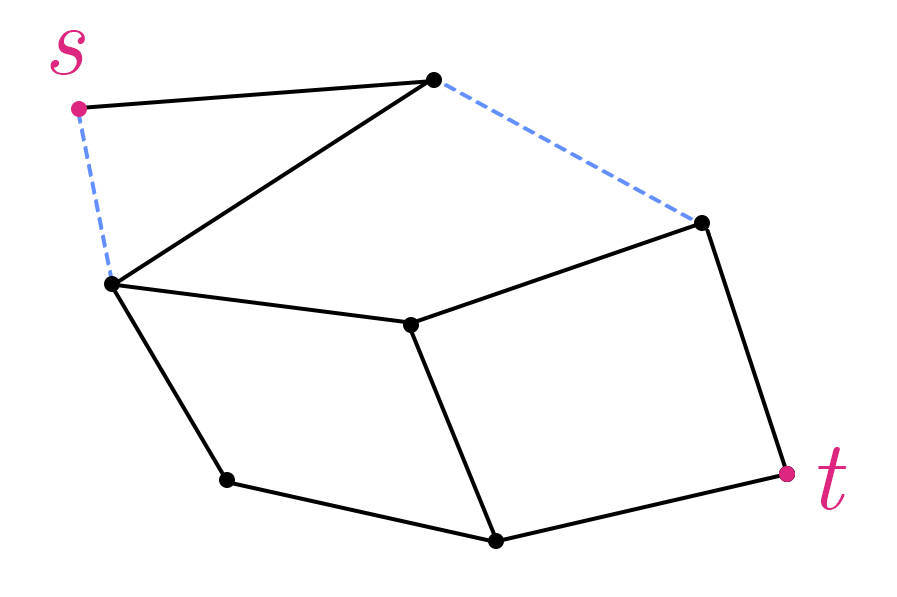}
\includegraphics[scale = 0.7]{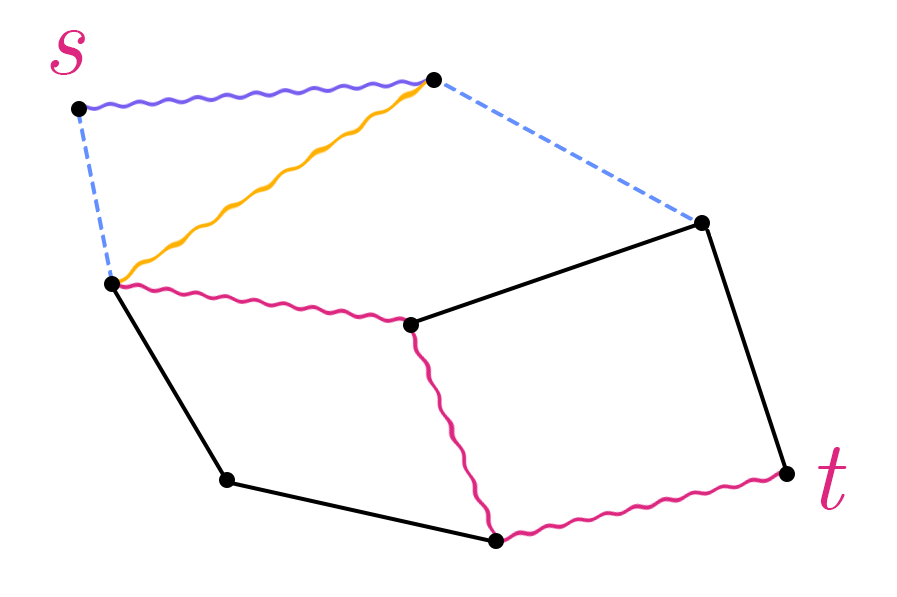}
\end{center}

\caption{\textbf{(Left)} Suppose that two blue dashed edges in the graph fail.  \textbf{(Right)} The restoration lemma guarantees that any shortest $s \leadsto t$ path (indicated with wavy edges) in the post-failure network can be partitioned into three subpaths, which are each a shortest path in the pre-failure network.}
\end{figure}

This restoration lemma suggests a natural approach for restoration algorithms: when $f$ edges fail and an $s \leadsto t$ shortest path is no longer usable, we can find a replacement $s \leadsto t$ shortest path by searching only over $s \leadsto t$ paths that can be formed by concatenating $f+1$ shortest paths that we have already computed in the current routing table.
Up to some subtleties involving shortest path tiebreaking \cite{tiebreaking}, this approach works, and has been experimentally validated as an efficient restoration strategy \cite{restoration}.
It has also found widespread theoretical application, e.g., in pricing algorithms \cite{HS01}, replacement path algorithms \cite{MR22, tiebreaking, BCGLPP18, MMG89, CC19, GJM20}, fault-tolerant variants of spanner problems \cite{tiebreaking, BGPV17, CCFK17}, and more.

\subsection{The Fault Tolerance/Subpath Count Tradeoff and First Main Result}

Most applications of the restoration lemma in the previous literature are actually based on the following generalization of Theorem \ref{thm:baserestoration}, which does not appear explicitly in \cite{restoration} but which follows easily from the main result:

\begin{corollary} [\cite{restoration}] \label{cor:baserestoration}
For any graph $G$ and any $1 \le k \le f$, every $f$-fault replacement path can be partitioned into at most $k+1$ subpaths that are each $(f-k)$-fault replacement paths in $G$.
\end{corollary}
\begin{proof} [Proof of Corollary \ref{cor:baserestoration}, given Theorem \ref{thm:baserestoration}]
Let $\pi$ be a replacement shortest path in a graph $G$ avoiding a set of edge failures $F = \{e_1, \dots, e_f\}$.
Consider the graph $G' := G \setminus \{e_1, \dots, e_{f-k}\}$.
In $G'$, $\pi$ is a $k$-fault replacement shortest path avoiding the remaining edge failures $\{e_{f-k+1}, \dots, e_f\}$.
Thus, applying Theorem \ref{thm:baserestoration} in $G'$, we can partition $\pi$ into $k+1$ subpaths that are each a shortest path in $G'$.
Each of these subpaths is an $(f-k)$-replacement shortest path, avoiding $\{e_1, \dots, e_{f-k}\}$, in the original graph $G$.
\end{proof}

An interpretation of this corollary is that there are two parameters of the decomposition that we might want to minimize: the \emph{number of subpaths} in the decomposition, and the \emph{complexity} of each subpath, as measured by the number of edges that would need to be excluded to make each subpath a shortest path.
Corollary \ref{cor:baserestoration} promises an \textbf{additive} tradeoff between subpath count and fault tolerance of each subpath.
The main contribution of this work is to show that the correct tradeoff is actually \textbf{multiplicative}:

\begin{theorem} [Main Result] \label{thm:intromain} The following hold for all positive integers $k, f$ with $k \le f$:
\begin{itemize}
\item \textbf{(Upper Bound)} In any graph $G$, every $f$-fault replacement path can be partitioned into at most $O(k)$ subpaths that are each a replacement path in $G$ avoiding at most $f/k$ faults. 

\item \textbf{(Lower Bound)} There are graphs $G$ and $f$-fault replacement paths $\pi$ that cannot be partitioned into $2k$ subpaths that are each an $(\lfloor f/k \rfloor - 2)$-fault replacement path in $G$.
\end{itemize}
\end{theorem}

The specific upper bound we show is $8k+1$ subpaths, although in this paper we do not focus on optimizing the leading constant.
We view Theorem \ref{thm:intromain} as a mixed bag, containing both good news and bad news for the area.
The good news is that the restoration lemma tradeoff has been substantially improved, and this potentially opens up new avenues for restoration algorithms for routing table recovery, as explored by Afek et al.~\cite{restoration}.
The bad news is that, in the important special case $k=1$ (i.e., decomposing into only two subpaths), our new lower bound shows that the previous restoration lemma was tight: there are examples in which one cannot decompose an $f$-fault replacement path into two replacement paths avoiding $f-2$ faults each (previously, a tradeoff of roughly $f/2$ was conceivable).
This case $k=1$ is particularly important in applications, especially to spanner and preserver problems \cite{tiebreaking, BGPV17}, and so this lower bound may close a promising avenue for progress on these applications.

\subsection{Weighted Restoration Lemmas and Second Main Result}

The original paper by Afek et al.~\cite{restoration} also proved a \emph{weighted} restoration lemma, which gives a weaker decomposition, but which holds also for weighted input graphs:
\begin{theorem} [Weighted Restoration Lemma \cite{restoration}]
For any \textbf{weighted} graph $G$ and any $1 \le k \le f$, every $f$-fault replacement path $\pi$ can be partitioned into $k+1$ subpaths and $k$ individual edges, where each subpath in the partition is an $(f-k)$-fault replacement path in $G$.
\end{theorem}

More specifically, this theorem promises that the subpaths and individual edges occur in an alternating pattern (although some of these subpaths in this pattern may be empty).
One can again ask whether this additive tradeoff between subpath count and fault tolerance per subpath is optimal.
We show that it is not, and that it can be improved to a multiplicative tradeoff, similar to Theorem \ref{thm:intromain}.

\begin{theorem} [Main Result, Weighted Setting] \label{thm:intromainwtd}
For any \textbf{weighted} graph $G$ and any $1 \le k \le f$, every $f$-fault replacement path $\pi$ can be partitioned into $O(k)$ subpaths and $O(k)$ individual edges, where each subpath in the partition is an $(f/k)$-fault replacement path in $G$.
\end{theorem}

For many graphs of interest, this theorem can be simplified.
For example, suppose we consider the setting of graphs that represent \emph{metrics}, in which we require that every edge is the shortest path between its endpoints.
Then we can consider the $O(k)$ individual edges in the decomposition to each be a $0$-fault replacement path, and so we could correctly state that $\pi$ can be partitioned into $O(k)$ subpaths that are each at most $(f/k)$-fault replacement paths.
However, we also note that our weighted main result cannot be simplified \emph{in general}: one can easily construct weighted graphs containing edges $(u, v)$ that are $f$-fault but not $(f-1)$-fault replacement paths between their endpoints, and therefore any weighted restoration lemma will need to include some exceptional edges in its hypotheses, as in \cite{restoration} and Theorem \ref{thm:intromainwtd}.

\subsection{Algorithmic Considerations}

The proofs of our new restoration lemmas (both weighted and unweighted) are given using a simple greedy decomposition strategy to find subpaths; essentially, we repeatedly peel off the longest possible prefix from the input path $\pi$ that is an $f/k$-fault replacement path, and then argue that this process will repeat at most $O(k)$ times.
Although this leads to a simpler proof, a downside of this greedy procedure is that it requires exponential time in the number of faults $f$.
That is, given a subpath $\pi_i \subseteq \pi$, it is not clear how to test whether $\pi_i$ is an $f/k$-fault replacement path, besides via brute force search over every subset of faults $F' \subseteq F, |F'| \le f/k$ which takes $\text{poly}(n) \cdot \exp(f)$ time.
To obtain algorithmic results, we thus need to change the decomposition strategy in Section \ref{sec:algorithms} entirely.
We show a more involved algorithm that implements our restoration lemmas in $\text{poly}(n, f)$ time.
That is:
\begin{theorem} [Unweighted Algorithmic Restoration Lemma] \label{thm:intromainalg}
There is an algorithm that take on input a graph $G$, a set $F$ of $|F|=f$ edge faults, a shortest path $\pi$ in $G \setminus F$, and a parameter $k$, and which returns:
\begin{itemize}
\item A partition
$\pi = \pi_0 \circ \pi_1 \circ \dots \circ \pi_q$
into $q=O(k)$ subpaths, and
\item Fault sets $F_0, \dots, F_q \subseteq F$ with each $|F_i| \le f/k$, such that each path $\pi_i$ in the decomposition is a shortest path in $G \setminus F_i$
\end{itemize}
(hence the algorithm implements Theorem \ref{thm:intromain}).
This algorithm runs in polynomial time in both the number of nodes $n$ and the number of faults $f$.
\end{theorem}

The core of of our new decomposition approach is a reduction to the algorithmic version of Hall's theorem; this is somewhat involved, and so we overview it in more depth in the next part of this introduction.
Using roughly the same algorithm, we also show the algorithmic restoration lemma in the weighted setting.

\begin{theorem} [Weighted Algorithmic Restoration Lemma] \label{thm:intromainalgwtd}
There is an algorithm that take on input a weighted graph $G$, a set $F$ of $|F|=f$ edge faults, a shortest path $\pi$ in $G \setminus F$, and a parameter $k$, and which returns:
\begin{itemize}
\item A partition
$\pi = \pi_0 \circ e_0 \circ  \dots \circ \pi_{q-1} \circ e_{q-1} \circ \pi_q,$
where each $\pi_i$ is a (possibly empty) subpath, each $e_i$ is a single edge, and $q=O(k)$, and
\item Fault sets $F_0, \dots, F_q \subseteq F$ with each $|F_i| \le f/k$, such that each path $\pi_i$ in the decomposition is a shortest path in $G \setminus F_i$
\end{itemize}
(hence the algorithm implements Theorem \ref{thm:intromainwtd}).
This algorithm runs in polynomial time in both the number of nodes $n$ and the number of faults $f$.
\end{theorem}

Although our algorithmic restoration lemmas run in polynomial time, they do not run in particularly efficient polynomial time, and optimizing this runtime is an interesting direction for future work.

\subsection{Technical Overview of Upper Bounds}

The more involved parts of the paper are the upper bounds in Theorem \ref{thm:intromain} and \ref{thm:intromainwtd}.
We will overview the proof in the unweighted setting (Theorem \ref{thm:intromain}) here; the weighted setting carries only a few additional technical details.

Let $\pi$ be an $f$-fault replacement path with endpoints $(s, t)$ in an input graph $G$.
In particular, let $F$ be a set of $|F| \le f$ edge faults, and suppose that $\pi$ is a shortest $s \leadsto t$ path in the graph $G \setminus F$.
We are also given a parameter $f' < f$, and our goal is to partition $\pi$ into as few subpaths as possible, subject to the constraint that each subpath is a replacement path avoiding at most $f'$ faults.

\paragraph{The Partition of $\pi$.}
We use a simple greedy process to determine the partition of $\pi$.
We will determine a sequence of nodes $(s=x_0, x_1, \dots, x_{k}, x_{k+1}=t)$ along $\pi$, which form the boundaries between subpaths in the decomposition.
Start with $s =: x_0$, and given node $x_i$, define $x_{i+1}$ to be the furthest node following $x_i$ such that the subpath $\pi[x_i, x_{i+1}]$ is an $f'$-fault replacement path.
We will denote the subpath $\pi[x_i, x_{i+1}]$ as $\pi_i$, and so the decomposition is
$$\pi = \pi_0 \circ \dots \circ \pi_k.$$
For each subpath we let $F_i \subseteq F, |F_i| \le f'$ be an edge set of minimum size such that $\pi_i$ is a shortest $x_i \leadsto x_{i+1}$ path in the graph $G \setminus F_i$.
(There may be several choices for $F_i$, in which case we fix one arbitrarily.)
Our goal is now to show that the parameter $k$, defined as (one fewer than) the number of subpaths that arise from the greedy decomposition, satisfies $kf' \le O(f)$. 

\paragraph{Proof Under Simplifying Assumptions.}

Our proof strategy will be to prove that an arbitrary faulty edge $e \in F$ can appear in only a constant number of subpath fault sets $F_i$, which implies that $kf' \le O(f)$ by straightforward counting.
To build intuition, let us see how the proof works under two rather strong simplifying assumptions:
\begin{itemize}
\item \textbf{(Equal Subpath Assumption)} We will assume that all subpaths in the decomposition have equal length: $|\pi_0| = \dots = |\pi_k|$.

\item \textbf{(First Fault Assumption)} Let us say that a \emph{shortcut} for a subpath $\pi_i$ is an alternate $x_i \leadsto x_{i+1}$ path in the original graph $G$ that is strictly shorter than $\pi_i$.
Every shortcut must contain at least one fault in $F_i$, and conversely, every fault in $F_i$ lies on at least one shortcut (or else it may be dropped from $F_i$).
Our second simplifying assumption is that for each $e \in F_i$, there exists a shortcut $\sigma$ for $\pi_i$ such that $e$ is the \emph{first} fault in $F_i$ on $\sigma$.
\end{itemize}

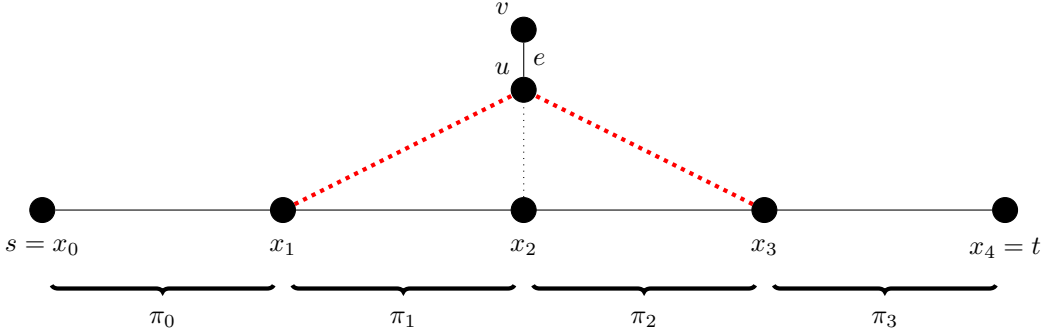
\begin{figure}[h]
\begin{center}
\begin{tikzpicture} [scale=0.8]
    \foreach \i in {1,...,5} {
        \node[draw, circle, fill] at ({4*\i},0) (node\i) {};
    }
    
    \draw (node1) -- (node5);
    
    \node[draw, circle, fill] at (12,2) (edgeStart) {};
    \node[draw, circle, fill] at (12,3) (edgeEnd) {};
    \node [above left of = edgeStart, node distance=0.4cm] {$u$};
    \node [above left of = edgeEnd, node distance=0.4cm] {$v$};
    
    \draw (edgeStart) -- node [midway, right] {$e$} (edgeEnd);
    
    \draw[dotted, red, ultra thick] (node2) -- (edgeStart);
    \draw[dotted] (node3) -- (edgeStart);
    \draw[dotted, red, ultra thick] (node4) -- (edgeStart);

    \node [below of = node1, node distance=0.5cm] {$s=x_0$};
    \node [below of = node2, node distance=0.5cm] {$x_1$};
    \node [below of = node3, node distance=0.5cm] {$x_2$};
    \node [below of = node4, node distance=0.5cm] {$x_3$};
    \node [below of = node5, node distance=0.5cm] {$x_4=t$};

    \node [below of = node1, node distance=1cm] (x1b) {};
    \node [below of = node2, node distance=1cm] (x2b) {};
    \node [below of = node3, node distance=1cm] (x3b) {};
    \node [below of = node4, node distance=1cm] (x4b) {};
    \node [below of = node5, node distance=1cm] (x5b) {};

    \draw [ultra thick, decoration={brace, mirror}, decorate] (x1b) -- node [midway, below=0.2cm] {$\pi_0$} (x2b);
    \draw [ultra thick, decoration={brace, mirror}, decorate] (x2b) -- node [midway, below=0.2cm] {$\pi_1$} (x3b);
    \draw [ultra thick, decoration={brace, mirror}, decorate] (x3b) -- node [midway, below=0.2cm] {$\pi_2$} (x4b);
    \draw [ultra thick, decoration={brace, mirror}, decorate] (x4b) -- node [midway, below=0.2cm] {$\pi_3$} (x5b);

    
\end{tikzpicture}
\end{center}
\caption{\label{fig:faultsetup} Under the equal subpath and first fault assumptions, we can reach contradiction if we assume that there are three different subpaths that all have shortcuts that use $e$ as their first fault.}
\end{figure}

With these two assumptions in hand, we are ready to prove that each faulty edge $e$ appears in only $O(1)$ many fault sets $F_i$.
Suppose for contradiction that there are three separate subpaths that all have shortcuts that use $e$ as their first edge, and moreover that these shortcuts use $e$ with the same orientation.
Consider the first and last of these shortcut prefixes, which we will denote as $q(x_1, u)$ and $q(x_3, u)$.
In Figure \ref{fig:faultsetup}, the shortcuts are represented as dotted paths, and $q(x_1, u), q(x_3, u)$ are colored red.
Notice that $q(x_1, u) \cup q(x_3, u)$ forms an alternate $x_1 \leadsto x_3$ path.
Since $e$ is assumed to be the first fault on these shortcuts, this alternate $x_1 \leadsto x_3$ path avoids all faults in $F$.
Additionally, by definition of shortcuts we have
$$|q(x_1, u)| + |q(x_3, u)| < \left|\pi_1 \right| + \left| \pi_3 \right|.$$
Since we have assumed that all subpaths have the same length, we can amend this to
$$|q(x_1, u)| + |q(x_3, u)| < \left|\pi_1 \right| + \left| \pi_2 \right|.$$
But this implies that $q(x_1, u) \cup q(x_3, u)$ forms an $x_1 \leadsto x_3$ path that is strictly shorter than the one used by $\pi$, which contradicts that $\pi$ is a shortest path in $G \setminus F$.
This completes the simplified proof, but the challenge is now to relax our two simplifying assumptions, which are currently doing a lot of work in the argument.

\paragraph{Relaxing the Equal-Subpath-Length Assumption.}

The equal-subpath-length assumption is the easier of the two to relax.
It is only used in one place in the previous proof: to replace $|\pi_3|$ with $|\pi_2|$ on the right-hand side of the inequality.
From this we see that $|\pi_2| \ge |\pi_3|$ is a good case for the argument, since this substitution remains valid.
The bad case is when $|\pi_2| < |\pi_3|$.

To handle this bad case, we follow a proof strategy from \cite{spanners}.
Let us say that a subpath is \emph{pre-light} if it is no longer than the preceding subpath, or \emph{post-light} if it is no longer than the following subpath.
In the above example, $\pi_2$ is pre-light if we have $|\pi_2| \le |\pi_1|$, and it is post-light if $|\pi_2| \le |\pi_3|$.
It is possible for a particular subpath to be both pre- and post-light, or for a particular subpath to be neither.
A simple counting argument shows that either a constant fraction (nearly half) of the subpaths are pre-light, or a constant fraction are post-light.
We will specifically assume in the following discussion that a constant fraction of the subpaths are post-light; the other case is symmetric.

We can now restrict the previous counting argument to the post-light subpaths only.
That is, we can argue that for each fault $e = (u, v)$ considered with orientation, there are only constantly many post-light subpaths for which it appears as the first fault of a shortcut.
The same counting argument then implies an upper bound of $|F_i| \le O(f/k)$ for the fault sets $F_i$ associated to post-light subpaths $\pi_i$, which completes the proof.

This still uses the first-fault assumption, and we next explain how this can be relaxed.
We consider the machinery used to relax the first-fault assumption to be the main technical contribution of this paper.

\paragraph{Relaxing the First-Fault Assumption.}

Let us now consider the case where there is a fault $e \in F_i$ that is \emph{not} the first fault of any shortcut for $\pi_i$.
We can still assume that there exists at least one shortcut $\sigma$ for $\pi_i$ with $e \in \sigma$ (otherwise, we can safely drop $e$ from $F_i$).
Let $e^*$ be the first fault along that shortcut $\sigma$.
We will shift the focus of our counting argument.
Previously, we considered each $(e \in F_i, \pi_i)$ as a pair, and our goal was to argue that faults $e$ can only be paired with a constant number of subpaths $\pi_i$.
Now, our strategy is to map the pair $(e \in F_i, \pi_i)$ to the different pair $(e^*, \pi_i)$, and our goal is to argue that each fault $e^*$ can only be paired with a constant number of subpaths $\pi_i$.
We call these new pairs $(e^*, \pi_i)$ \emph{Fault-Subpath (FS) Pairs}, and we formally describe their generation in Section \ref{sec:fspairs}.
(We note that, for a technical reason, we actually generate FS pairs using \emph{augmented subpaths} that attach one additional node to $\pi_i$ - but to communicate intuition about our proof, we will ignore this detail for now.)

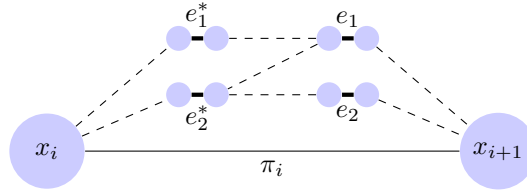
\begin{figure}
\begin{center}
\begin{tikzpicture}
    \node [style={circle, fill=blue!20, minimum size=1cm}] at (0,0) (x_i) {$x_i$};
    \node  [style={circle, fill=blue!20, minimum size=1cm}] at (6,0) (x_ip1) {$x_{i+1}$};

    \draw (x_i) --node [midway, below] {$\pi_i$} (x_ip1);


    \node [style={circle, fill=blue!20}] (A) at (3.75, 1.5) {};
    \node [style={circle, fill=blue!20}] (B) at (4.25, 1.5) {};
    \node [style={circle, fill=blue!20}] (A2) at (3.75, 0.75) {};
    \node [style={circle, fill=blue!20}] (B2) at (4.25, 0.75) {};
    \node [style={circle, fill=blue!20}] (A3) at (1.75, 1.5) {};
    \node [style={circle, fill=blue!20}] (B3) at (2.25, 1.5) {};
    \node [style={circle, fill=blue!20}] (A4) at (1.75, 0.75) {};
    \node [style={circle, fill=blue!20}] (B4) at (2.25, 0.75) {};
    
    \draw[ultra thick] (A) to node [midway, above] {$e_1$} (B);
    \draw[ultra thick] (A2) to node [midway, below] {$e_2$} (B2);
    \draw[ultra thick] (A3) to node [midway, above] {$e^*_1$} (B3);
    \draw[ultra thick] (A4) to node [midway, below] {$e^*_2$} (B4);

    \draw [dashed] (B) -- (x_ip1);
    \draw [dashed] (B2) -- (x_ip1);
    \draw [dashed] (x_i) -- (A3);
    \draw [dashed] (x_i) -- (A4);

    \draw [dashed] (A) -- (B3);
    \draw [dashed] (A2) -- (B4);
    \draw [dashed] (A) -- (B4);

\end{tikzpicture}
\end{center}
\caption{In order to relax the first fault assumption, instead of counting $(e_1, \pi_i)$ and $(e_2, \pi_i)$ as pairs, we can map these to distinct FS pairs $(e^*_1, \pi_i), (e^*_2, \pi_i)$.  Our main technical step is to show that this distinct mapping is always possible.}
\end{figure}

Although we can bound the number of FS pairs $(e^*, \pi_i)$ as before, this only implies our desired bound on the size of the fault sets $|F_i|$ if we can \emph{injectively} map each pair $(e \in F_i, \pi_i)$ to a \emph{distinct} FS pair $(e^*, \pi_i)$.
The main technical step in this part of the proof is to show that this injective mapping is always possible.
Let $\Gamma_i$ be a bipartite graph between vertex sets $F_i$ and $F$.
Put an edge between nodes $e \in F_i, e^* \in F$ if and only if there exists a shortcut $\sigma$ for $\pi_i$, in which $e \in \sigma$ and $e^*$ is the first fault in $\sigma$.
An injective mapping to FS pairs corresponds to a matching in $\Gamma_i$ of size $|F_i|$, i.e., a matching of maximum possible size (given that one of the sides of the bipartition has only $|F_i|$ nodes).
The purpose of this graph construction is to enable the following application of Hall's theorem:

\begin{lemma} [Hall's Theorem]
The following are equivalent:
\begin{itemize}
\item The graph $\Gamma_i$ has a matching of size $|F_i|$.
(Equivalently, one can associate each pair $(e \in F_i, \pi_i)$ to a \textbf{unique} FS pair.)

\item There does not exist a subset of faults $F'_i \subseteq F_i$ whose neighborhood in $\Gamma_i$ is strictly smaller than $F'_i$ itself (that is, $|N(F'_i)| < |F'_i|$). 
\end{itemize}
\end{lemma}

We show that the latter property is implied by minimality of $F_i$, which means the former property is true as well.
See Lemma \ref{lem:hallreduce2} and surrounding discussion for details.

\section{Preliminaries}

\begin{definition}
    Relative to a value of $f$, we call the pair $(q, r)$ \textbf{restorable} if in every graph $G$, any $f$-fault replacement path can be partitioned into $q$ subpaths which are each $r$-fault replacement subpaths in $G$.
\end{definition}

\begin{remark}
    (Monotonicity of Restorability) If $(q, r)$ is restorable, then both $(q+1, r)$ and $(q, r+1)$ are restorable. Equivalently, if $(q, r)$ is not restorable, neither $(q-1, r)$ nor $(q, r-1)$ are restorable.
\end{remark}

For notation, we will commonly write $F$ for a set of $f$ failing edges, and $s, t$ for the endpoint nodes of the replacement path under consideration, which is then written $\pi(s,t \mid F)$.
To denote a subpath between intermeidate nodes $u$ and $v$, we write $\pi(s,t \mid F)[u,v]$.
We will denote the $q$ $r$-fault replacement subpaths as $\pi(x_i, x_{i+1} \mid F_{i+1})$ with $x_0 := s$ and $x_q := t$, and each fault subset $|F_i| \leq r$. Then a decomposition of $\pi(s,t \mid F)$ satisfying restorability can be written 
\[\pi(s,t\mid F) = \pi(x_0, x_1\mid F_1) \circ \pi(x_1, x_2\mid F_2) \circ \ldots \circ \pi(x_{q-1}, x_q\mid F_q).\]
Equivalently, for each $i$, 
\[\pi(s,t \mid F)[x_{i}, x_{i+1}] = \pi(x_{i}, x_{i+1} \mid F_{i+1}).\]

\section{Lower Bounds}\label{lower}

As a warmup, we begin by showing our lower bound against decomposition into $2$ faults.
\begin{proposition}\label{2subpath}
For all $f \geq 2$, $(2, f-2)$ is not restorable.
\end{proposition}

\begin{figure}[h]
    \centering
    \includegraphics[scale = 0.5]{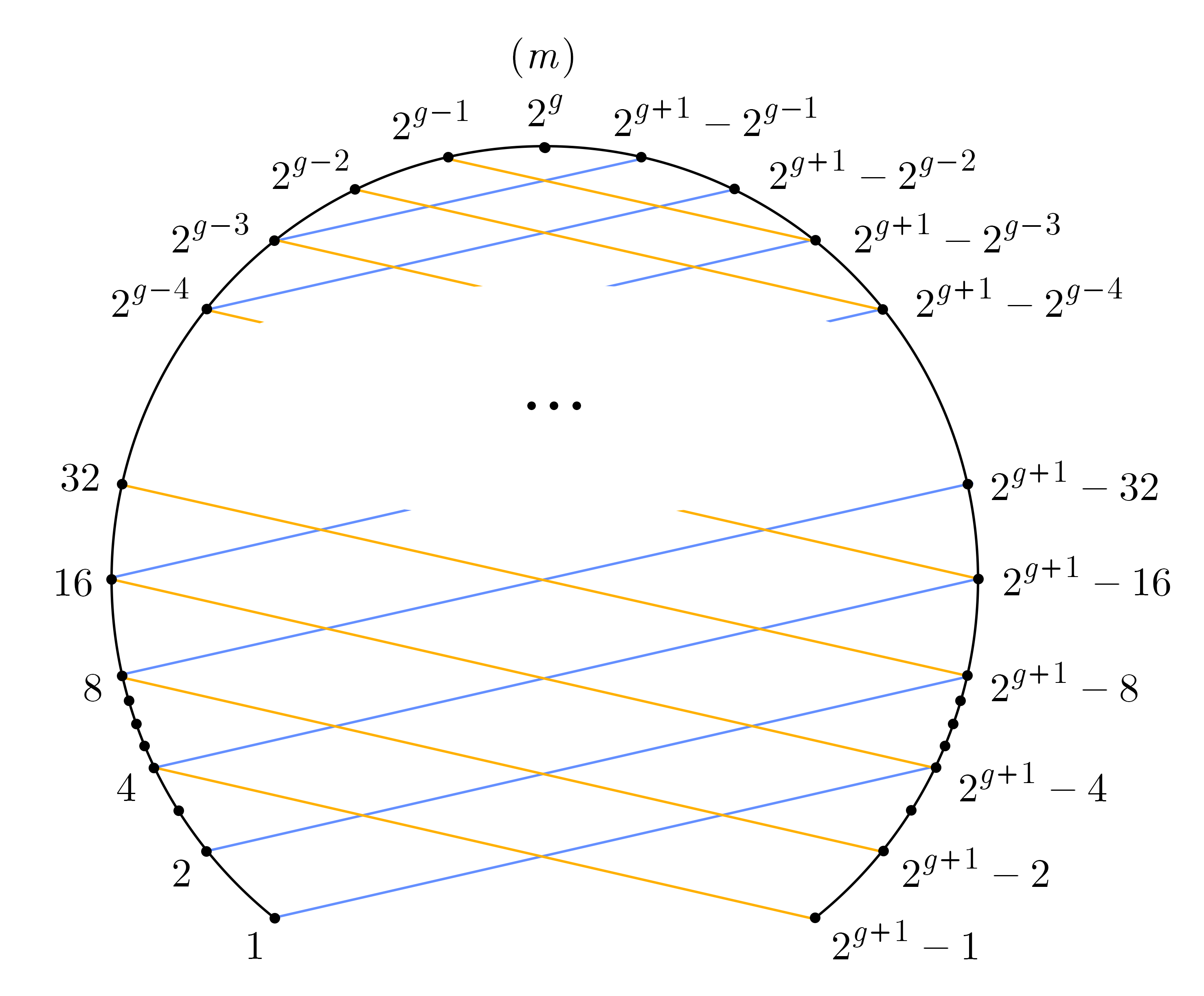}
    \caption{If the blue and yellow edges fail (i.e.\ all straight-line edges on the inside of the outer semicircle), then we can't partition the remaining shortest path (black edges along the outer semicircle) into two subpaths that are both $(f-2)$-fault replacement paths. For clarity, only power of two vertices are drawn here, but the outer cycle contains $2^{g+1}-1$ vertices.}
    \label{fig:lowerboundsimple}
\end{figure}
\begin{proof}
We will first assume for convenience that $f$ is even, and return to the case where $f$ is odd at the end.
Let $g = f/2$ and 
let $G_f$ be the graph as illustrated in Figure \ref{fig:lowerboundsimple}.
Formally: the vertices of $G_f$ are $1, 2, \ldots , N:=2^{g+1}-1$ (labeled clockwise in Figure \ref{fig:lowerboundsimple}),
and its edge set is $E_1 \cup E_2 \cup E_3$, where 
\[E_1 := \{(2^k, 2^{g+1}-2^{k+2}), 0 \leq k \leq g-3\}\]
\[E_2 := \{(2^{k+2}, 2^{g+1}-2^{k}), 0 \leq k \leq g-3\}\]
\[E_3 := \{(i, i+1), 1 \leq i \leq N-1\}\]
In Figure \ref{fig:lowerboundsimple}, the edges in $E_1$ are drawn in blue and slope upwards to the right, and the edges in $E_2$ are drawn in yellow and slope upwards to the left. $E_3$ is in black and forms the outer curve. 
Let $F := E_1 \cup E_2$, and notice there is a unique replacement path $\pi(1, N \mid F)$ which consists of $E_3$, the outer curve. Note that $G_f$ is symmetric about the vertex $m = 2^g$, which is also the midpoint of $\pi(1, N \mid F)$.
Define the ``half-arcs'' of this graph as $\pi(1, m\mid F)$ and $\pi(m, N\mid F)$, the two subpaths partitioning $\pi(1,N\mid F)$ into equal parts divided at midpoint $m$.
(We note that this partitioning will be used again in Lemma \ref{lem:halfarc} as well.)

Let $x \in \pi(1, N \mid F)$ be an arbitrary vertex, which splits the path into a prefix from $1$ to $x$ and a suffix from $x$ to $N$.
Let $F_1, F_2$ respectively be minimum-size fault sets such that the prefix and suffix are replacement paths avoiding $F_1, F_2$.
We will write the prefix and suffix as
$$\pi(1, x\mid F_1) \quad \text{and} \quad \pi(x,N\mid F_2).$$
By symmetry of the construction we may assume without loss of generality that $x \ge m$, and so
$$\pi(1, m \mid F) \subseteq \pi(1, x \mid F_1).$$
We will now proceed to show that $F_1$ must contain every edge in $E_1$ and all but one edge in $E_2$, and hence $|F_1| \geq f - 1$.

\paragraph{First Part (Proof of $E_1 \subseteq F_1$).} 

Consider $\pi(1,x\mid F_1)$; suppose for a contradiction that there is an edge $(2^c, 2^{g+1}-2^{c+2}) \in E_1 \setminus F_1$, with $c \leq g-3$. Then we can construct a path $p$, a shortcut from $1$ to $x$ by traversing through $(2^c, 2^{g+1}-2^{c+2})$ and then using edges in $E_3$ to get to $x$. Explicitly, this path is:
\[p=\left\{\begin{array}{ll}
(1,2, \ldots 2^c) \circ (2^c, 2^{g+1}-2^{c+2}) \circ (2^{g+1}-2^{c+2}, \ldots x+1, x) & \text{ if } x \leq 2^{g+1}-2^{c+2}\\
(1,2, \ldots 2^c) \circ (2^c, 2^{g+1}-2^{c+2}) \circ (2^{g+1}-2^{c+2}, \ldots x-1, x) & \text{ if } x > 2^{g+1}-2^{c+2}.
\end{array}\right.\]
In the first case, the length of $p$ is 
\[|p| = 2^c + 2^{g+1} - 2^{c+2} - x.\]
Since by assumption $x \geq m = 2^g$, we have $2^{g+1} - x \leq 2^g \leq x$, so we can upper bound 
\[|p| \leq 2^c - 2^{c+2} + x.\]
In the second case, the length of the path is
\[|p| = 2^c+2^{c+2} - 2^{g+1} + x.\]
Since $c \leq g-3$, we can directly upper bound the path length as 
\[|p| \leq 2^{g-3} + 2^{g-1} - 2^{g+1} + x.\]
In either case, the length of $p$ is strictly less than $x-1$, the length of shortest path $\pi(1,x \mid F)$, a contradiction. We must therefore have $E_1 \subseteq F_1$.

\paragraph{Second Part (Proof of $|E_2 \setminus F_1| \le 1$):}
Suppose for a contradiction that there are two edges of $E_2$ which $F_1$ does not contain: $(2^{a+2}, 2^{g+1}-2^{a})$ and $(2^{b+2}, 2^{g+1}-2^{b})$ with $a<b$. Then in $G \setminus F_1$ we have a walk $w$ from $1$ to $x$ defined as

\begin{align*}
w = (1, 2, \ldots 2^{a+2}) &\circ (2^{a+2}, 2^{g+1}-2^{a})\circ (2^{g+1}-2^{a}, 2^{g+1}-2^{a}-1, \ldots 2^{g+1}-2^{b})\\
&\circ (2^{g+1}-2^{b}, 2^{b+2}) \circ (2^{b+2}, 2^{b+2}+1, \ldots, x-1, x),
\end{align*}
of length
\begin{align*}
    |w| &=2^{a+2} -2^a + 2^{b} + 1 + x - 2^{b+2}\\
    &= x - 3(2^b - 2^a) + 1
\end{align*}
Then the length of $w$ will be strictly less than $x - 1$, the length of $\pi(1,x \mid F)$, a contradiction. Thus we must include all of $F$ in $F_1$ except at most one edge from $E_2$. 

Finally, in the case that $f$ is odd, we instead construct $G_f$ with $g = \lceil f/2\rceil$, and take any edge out of $E_1$ or $E_2$, which does not change the analysis.
\end{proof}

Our lower bound with two subpaths generalises to our main lower bound result, which we rewrite below:
\begin{proposition}
For any $k \in \nn$,
$(2k, \lfloor {f}/{k} \rfloor - 2)$ is not restorable.

\end{proposition}
\begin{proof}
Assume for convenience that $k$ divides $f$. We will glue $k$ copies of the graph with $f/k$ faults in the previous proposition together, and then show that for any division of a particular $f$-fault replacement path into subpaths, one subpath must contain one of the half-arcs as defined before, and its fault set will have to include $f/k - 1$ faults.
\newline
\newline
We take $k$ copies of $G_{f/k}$ as defined in Proposition \ref{2subpath}, denoted by $G_{1, f/k}, G_{2, f/k}, \ldots G_{k, f/k}$, labeling the vertices of $G_{i, f/k}$ as $(i, j)$ where $j$ is the label of the corresponding vertex in $G_{f/k}$. We identify each $(i, 2^{g+1}-1)$ with $(i+1, 1)$. 
The edges in this graph are the union of all edges of the $G_{i, f/k}$ (see Figure \ref{fig:lowerboundmulti}), and we define $F$ as the union of the fault sets of each $G_{f/k}$ as defined in the proof of Proposition \ref{2subpath}.
\begin{figure} [t]
\begin{center}
    \includegraphics[scale = 1]{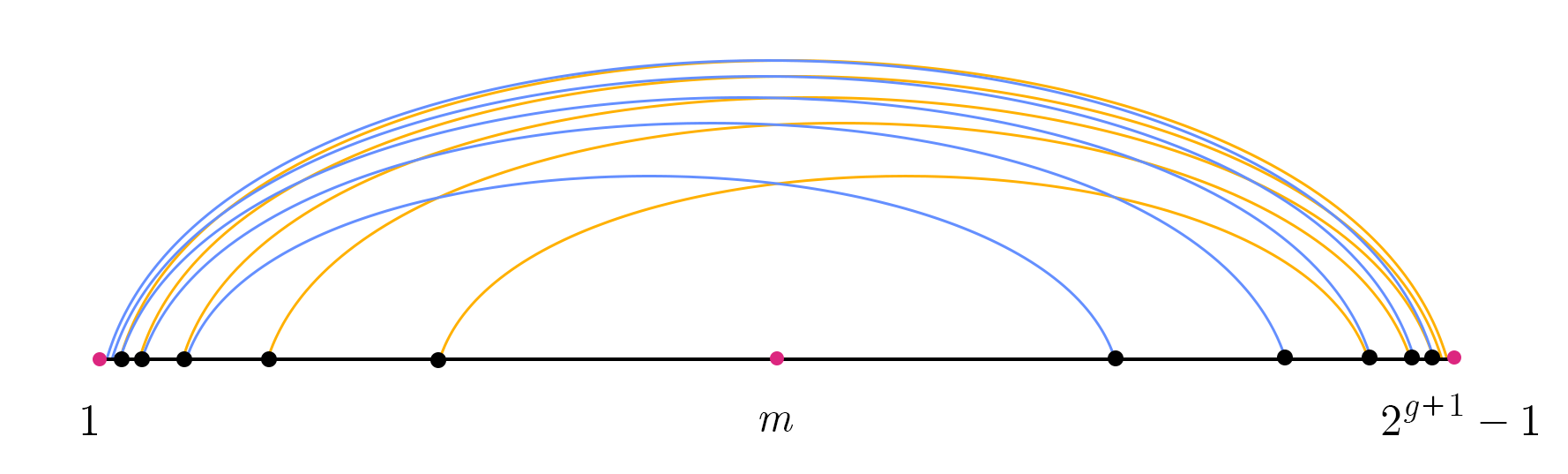}
\end{center}
\begin{center}
    \includegraphics[scale = 0.25]{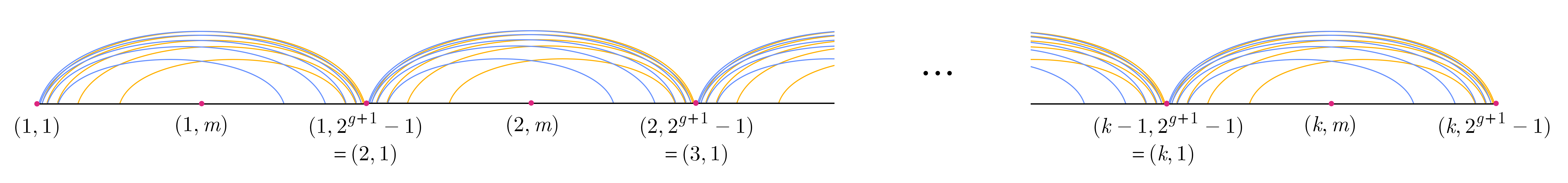}
    \caption{\label{fig:lowerboundmulti} The top figure depicts one copy of $G_{f/k}$, and the bottom depicts all the copies combined together.}
\end{center}
\end{figure}
Let $E_{j,i}$ denote the $E_j$ for $G_{i, f/k}$ with $j \in \{1,2\}$, so that formally \[F := \bigcup_{i=1}^{k} \big( E_{1,i} \cup E_{2, i} \big).\] 
Let $s := (1,1)$, $t := (k, 2^{g+1}-1)$. Consider $\pi(s, t\mid F)$. 
This $f$-fault replacement path is precisely the non-fault edges in $G$, or the union of $E_{3,i}$ over each of the $G_{i,f/k}$.

We now bring in the previous half-arc structure from the case with two subpaths. This graph contains all the half-arcs of each $G_{i,f/k}$, and the half-arcs can be expressed either as $\pi(s,t)[(i,1),(i,m)]$ or $\pi(s,t)[(i,m),(i,2^{g+1}-1)]$. From Proposition \ref{2subpath}, we have the following:
\begin{lemma}\label{lem:halfarc}
    A path containing a half-arc of any $G_{f/k}$ subgraph cannot be a $(f/k-2)$-fault replacement path.
\end{lemma}

\begin{proof}
    Following from the argument of Proposition \ref{2subpath}, a fault replacement path containing a half arc of $G_{i,2f/k}$ must have its fault set contain at least every edge in $E_{1,i} \cup E_{2,i}$ except possibly one. Thus any fault set of that path has size at least $f/k -1$.
\end{proof}

We will show that any division of $\pi(s,t \mid F)$ into $2k$ subpaths will result in one subpath containing a half-arc, and thus failing to be a $(f/k - 2)$-fault replacement path. Suppose we have some choice of boundary vertices $x_1, x_2, \ldots x_{2k-1}$ and corresponding fault subsets $F_1, F_2, \ldots F_{2k}$, so that each $\pi(s,t)[x_{i-1}, x_{i}]$ is a shortest path in $G \setminus F_i$.

Let the \emph{interior vertices} of a path denote all its vertices except its first and last. Note that $\pi(s,t \mid F)$ contains $2k$ half-arcs, and any half-arc which does not have any $x_i$ in its interior vertices will be completely contained in some $\pi(s,t)[x_{j-1}, x_{j}]$. The interior vertices of all $2k$ half-arcs are disjoint, and we only have $2k-1$ $x_i$ which can be in the interior of half arcs. Therefore some subpath $\pi(s,t)[x_{i-1}, x_{i}]$ must contain a half-arc, and its fault set $|F_i|$ must have size at least $f/k -1$. Thus we will always get that one of the subpaths cannot be a $(f/k - 2)$-fault replacement subpath, proving the lower bound.

In the case when $k$ does not divide $f$, we choose graphs which are as even as possible to combine; Let $a$ be the remainder of $f$ divided by $k$. We glue $a$ copies of $G_{\lfloor f/k \rfloor + 1}$ to $(k-a)$ copies of $G_{\lfloor f/k \rfloor}$. In this case the subpath which contains a half-arc might contain a half arc of $G_{\lfloor f/k \rfloor}$, and will enforce a fault set of size only $\lfloor f/k \rfloor - 1$.
\end{proof}

If we want a similar result using this method for the case for an odd number of subpaths, say $2k-1$, we still need to construct $k$ copies of $G_{f/k}$, since half-arcs come in pairs, and we get the same bound on fault sets. Alternatively, we can also use monotonicity to directly get:

\begin{corollary}
    For any $k \in \nn$, $(2k-1, \lfloor {f}/{k} \rfloor - 2)$ is not restorable.
\end{corollary}

\section{Upper Bound}

We now prove Theorem \ref{thm:intromain}.
Fix any $s,t$ and replacement path $\pi(s,t \mid F)$, where $|F|=: f$. Recall that, to prove Theorem \ref{thm:intromain}, our goal is to show that for any $k \in \rr$, we can partition $\pi(s,t \mid F)$ into $O(k)$ $(f/k)$-fault replacement subpaths. 

\subsection{Subpath Generation}
    We generate $\{x_i\}$, the vertices used to split $\pi(s,t \mid F)$, by traversing along $\pi(s,t\mid F)$ and adding vertices greedily to the current subpath until adding one more vertex would make that subpath no longer an $f/k$-fault replacement path. More precisely, we set $x_0 := s$, and then we pick each $x_{i+1}$ to maximize $|\pi(s,t \mid F)[x_{i}, x_{i+1}]|$ under the constraint that $\pi(s,t \mid F)[x_{i}, x_{i+1}]$ is an $f/k$-fault replacement subpath. We define $x_{i}'$ to be the vertex immediately after $x_i$ along $\pi(s,t \mid F)$, i.e., the vertex satisfying 
    \[|\pi(s,t \mid F)[x_i, x_i']| = 1.\]
    Let $q$ be the number of subpaths generated by this process, and so we have $\{x_1, \dots, x_q\}$.
    We will denote the $q$ corresponding subpaths of $\pi(s,t\mid F)$ by 
\begin{equation*}
\pi_i= \pi(s,t \mid F)[x_{i}, x_{i+1}] \quad \forall 0 \leq i \leq q-1.
\end{equation*} 
\noindent Here we remark that the last subpath $\pi_{q-1}$ differs from the others as it is bounded by the end of $\pi(s,t \mid F)$ and is not generated greedily.

\begin{definition} [Augmented Subpaths]
For $0 \leq i \leq q-2$, We define \textbf{augmented subpaths} $\pi_i'$ as \[\pi'_i := \pi(s,t\mid F)[x_i, x_{i+1}'].\]  Note that no augmented subpath $\pi_i'$ can be an $f/k$-fault replacement subpath in $G$ by greedy choice of $x_i$. 
\end{definition}

For each $i$, fix $F_i$ to be any minimum size fault set such that $\pi'_i$ is a shortest path in $G \setminus F_i$. That is, $\pi'_i$ is an $|F_i|$-fault replacement path but not a $c$-fault replacement path for any $c < |F_i|$.  By our choice of $x_i$, we must have \[|F_i| > f/k.\]

Following the notation of \cite{spanners}, we will denote $\pi_i'$ as \textbf{pre-light} if its length is less than or equal to the length of $\pi_{i-1}'$, and \textbf{post-light} if its length is less than or equal to the length of $\pi_{i+1}'$. From Lemma 3 of \cite{spanners}, at least half of the $\pi_i'$ are pre-light, or at least half are post-light. We will assume without loss of generality that for at least $\frac{q-1}{2}$ $i$, $\pi_i'$ is post-light.
The other case, where at least half of the $\pi_i'$ are pre-light, follows from a symmetric argument.\footnote{In particular, in the case where at least half of the $\pi_i'$ are pre-light, one can use the following argument but substitute ``left ends'' for ``right ends", and ``(left) FS-pairs'' for ``(right) FS-pairs''.}

\subsection{FS-pair Generation \label{sec:fspairs}}
To bound the number of subpaths, we will use a counting argument tracking pairs of augmented subpaths and faults which are ``close'' to each other. In particular, any subpath-fault pair has a fixed distance to each other within the fault-free graph, and only a limited number of subpaths can have small distance to a specified fault without creating a shorter path than $\pi(s,t\mid F)$. We will formalize these ideas using FS-pairs in this section.
\begin{definition}
    For any $(u,v)$-path $p'$, we say a $(u,v)$-path $p$ is a \textbf{shortcut} of $p'$ if $\text{len}(p) < \text{len}(p')$.\footnote{We include the possibility of non-simple shortcuts, which may repeat nodes and be walks.  Our existential upper bound proof would work equally well if we restricted attention to \emph{simple} shortcuts, but this expanded definition will be more convenient for algorithmic reasons outlined in Section \ref{sec:algorithms}.}
\end{definition}
\begin{definition}
    A (right) FS-pair is a pair $(e^*, \pi_i')$ with $e^* \in F$, $\pi_i'$ a post-light augmented path, and the property that there exists a fault-free path from $x'_{i+1}$ to $e^*$ which is contained in a shortcut of $\pi_i'$.
\end{definition}
We assume above that most subpaths are post-light. In the other case when most subpaths are pre-light, we would instead define left FS-pairs, where we require $\pi_i'$ to be pre-light and the fault-free path from $x_i$ to $e^*$, and the argument would be handled symmetrically. 

We will now describe the generation of right FS-pairs, starting by setting up some notation. Fix any post-light augmented subpath $\pi_i'$. For a given $e \in F_i$ which we refer to as the \textbf{generating fault}, let $S_{e}$ be the set of $(x_{i+1}',x_{i})$-shortcuts of $\pi_i'$\footnote{We consider $\pi_i$ a $(x_{i+1}',x_{i})$-path in this section.} which contain $e$.\footnote{Note that $S_{e}$ depends on the choice of subpath $\pi_i'$, although we do not include this parameter in the notation.}
For a shortcut $p \in S_{e}$, we define the \textbf{base fault} $b(p) \in F$ of $p$ as the first fault in $p$ traversing from $x_{i+1}'$ to $x_{i}$. More precisely, we define
\[b(p) := e_{\min\{j: e_j \in F\}} \text{ where } p = x_{i+1}'e_1 v_1 e_2 \ldots e_m x_{i}.\]
For each $e \in F_i$, we define its set of base faults for $\pi_i'$ as \[B(e) := \{b(p): p \in S_{e}\}\]
Finally, we define the family of base faults for $\pi_i'$ as
\[\mathcal{B}(\pi_i') = \{B(e): e\in F_i\}\]
Our next goal will be to choose a distinct base fault from each $B(e) \in \mathcal{B}(\pi_i')$ in order to construct at least $|F_i|$ FS-pairs with $\pi_i'$.

Explicitly, we define an auxiliary bipartite graph $\Gamma_i$ where one side of the bipartition is $F_i$ and the other is the set of faults $F$. For $e \in F$ and $e^* \in F_i$, we have $(e, e^*) \in E(\Gamma_i)$ if and only if $e^* \in B(e)$. In this set up, choosing a distinct base fault for each $B(e)$ is equivalent to finding a matching of $\Gamma_i$ which saturates $F_i$.
Hall's Theorem gives us a condition for this:

\begin{lemma} [Hall's Condition] \label{lem:hall}
If for every $A \subseteq F_i$ we have $|N(A)| \geq |A|$, where $N(A)$ is the neighborhood of $A$ in $\Gamma_i$, then $\Gamma_i$ contains a matching that saturates $F_i$ (and therefore it is possible to choose a distinct base fault for each $B(e)$).
\end{lemma}

We therefore only need to verify the premises of Hall's Condition.  The following lemma will be helpful.

\begin{lemma} \label{lem:hallreduce2}
For any $A \subseteq F_i$, the fault set $F'_i := (F_i \setminus A) \cup N(A)$ is also a valid fault set for $\pi_i'$.
(That is, $\pi_i'$ is a shortest path in $G \setminus F'_i$.)
\end{lemma}
\begin{proof}
First, we observe that
$$N(A) = \left|\bigcup_{ e \in A} B(e) \right|.$$
This holds because the neighbours of each generating fault $e$ is the set of base faults $B(e)$ it generates.

We now need to prove that no shortcuts for $\pi_i'$ survive in $G \setminus F'_i$. Let $p$ be an arbitrary shortcut for $\pi_i'$ in $G$.
Then it must contain some fault $e' \in F_i$, since $F_i$ is a valid fault set.
There are two cases:
\begin{itemize}
\item If $e' \in F'_i$, then the shortcut $p$ does not survive in $G \setminus F'_i$.

\item Otherwise, suppose that $e' \notin F'_i$, and so in particular $e' \in A$.
In this case, $p$'s base fault $b(p)$ is in $B(e') \subseteq F'_i$, and thus not in $G \setminus F'_i$. 
\end{itemize}
Therefore there are no surviving shortcuts for $\pi_i'$ in $G \setminus F'_i$.
\end{proof}

Notice that Lemma \ref{lem:hall} follows from Lemma \ref{lem:hallreduce2}: since we assume that $F_i$ is a \emph{minimal} fault set, we must have that $|N(A)| \ge A$ for all $A \subseteq F_i$, since otherwise we would have $|F'_i| < |F_i|$.
Since Hall's condition holds, over any augmented subpath $\pi_i'$, we can assign a unique base fault to every generating fault. Accordingly, we can define an injective function $\phi_i: F_i \to F$ where $\phi_i(e) \in B(e)$.
\newline
\newline
We will construct our FS-pairs for $\pi_i'$ as $\{(\phi_i(e), \pi_i') \mid e \in F_i \}$, and since $\phi_i$ is injective, we get $|F_i|$ distinct FS-pairs from $\pi_i'$. We repeat this process for every post-light augmented subpath. It follows that we will generate at least $(q-1)f/(2k)$ FS-pairs, since we have $(q-1)/2$ post-light augmented subpaths which have corresponding fault sets $F_i$ each with at least $f/k$ faults. 

\subsection{Analysis of FS-pairs}

\begin{lemma}\label{fsc}
    Each fault in $F$ will be in at most 4 FS-pairs. 
\end{lemma}
\begin{proof}
Recall that only post-light $\pi_i'$ will be in FS-pairs. Suppose, for a contradiction, that there is some fault $e = (u,v)$ associated with 5 $\pi_i'$ in FS-pairs. By pigeonhole, at least 3 of the $\pi_i'$ have fault-free paths (as subsets of some shortcut) from their right ends $x_{i+1}'$ to $u$, or at least 3 of the $\pi_i'$ have fault-free paths from their right ends to $v$. Without loss of generality assume this is $u$. Let these subpaths be $\pi_a'$, $\pi_b'$, and $\pi_c'$, with $a < b < c$. We will also label the fault-free paths as $p_a$, $p_b$, and $p_c$. We have \[|p_a| \leq |\pi_a'| - 2 \quad \text{and} \quad |p_c| \leq |\pi_c'| - 2\] since the shortcut of $\pi_a'$ which $p_a$ is on has length at least $|p_a| + 1$ when we include $e$, and same with $p_c$. 
\newline
\newline
Since $\pi_a'$ is post-light, we have \[|\pi_{a+1}'| \geq |\pi_a'|. \]
With each $\pi_i'$ being extended from $\pi_i$ by one vertex, we have also \[|\pi_{a+1}| \geq |\pi_a|.\]
Moreover, since $a < b < c$, $a+1 \neq c$. Note that the distance from $x_{a+1}'$ to $x_{c+1}'$ in $G \setminus F$ is equal to their distance along $\pi(s,t|F)$, a shortest path they're both on, which gives us a lower bound of
\begin{align*}
    d_{G \setminus F}(x_{a+1}', x_{c+1}') &= \pi(s, t \mid F)[x_{a+1}', x_{c+1}'] \\
    &= \sum_{i = a+1}^c |\pi_i|\\
    &\geq |\pi_{a+1}| + |\pi_{c}| \\
    &\geq |\pi_{a}| + |\pi_{c}|.
\end{align*}
However, $p_a$ and $p_c$ give a fault-free path from $x_{a+1}'$ to $ x_{c+1}'$ in $G\setminus F$ also, which upper bounds their distance as 
\begin{align*}
    d_{G \setminus F}(x_{a+1}', x_{c+1}') &\leq d_{G \setminus F}(x_{a+1}', u) + d_{G \setminus F}(u, x_{c+1}')\\
    &\leq |p_a| + |p_c| \\
    &\leq |\pi_a'| + |\pi_c'| - 4 \\
    &= |\pi_{a}| + |\pi_{c}| - 2.    
\end{align*}
Which contradicts $\pi(s,t \mid F)$ being a shortest path. Therefore, each fault in $F$ is associated with at most 4 $\pi_i'$ over all FS-pairs. 
\end{proof}

\noindent
We are now ready to finish the proof of Theorem \ref{thm:intromain}.
We can generate at least $\frac{(q-1)f}{2k}$ FS-pairs, but each fault is in at most 4 FS-pairs, and there are only $f$ faults, so we have
\[4f \geq \frac{(q-1)f}{2k}.\]
Rearranging, we can upper bound $q$, the number of subpaths as \[q \leq 8k + 1.\]
  
\begin{corollary} \label{cor:fsbound}
For any partition of $\pi(s,t \mid F)$ into subpaths $\pi_i$, there are at most $4f$ right FS-pairs containing post-light augmented subpaths $\{\pi_i'\}$. 
\end{corollary}

Again, in the other case where most subpaths are pre-light, the relevant corollary is that there are at most $4f$ left FS-pairs containing pre-light augmented subpaths $\{\pi_i'\}$.
The proof is essentially identical.

\section{Weighted Upper Bound}

We next prove Theorem \ref{thm:intromainwtd}.
Recall that the goal is to prove that in any weighted graph $G$, every $f$-fault replacement path $\pi$ can be partitioned into
$$\pi = \pi_0 \circ e_0 \circ \pi_1 \circ e_1 \circ \dots \circ e_{q-2} \circ \pi_{q-1}$$ where each $e_i$ is an edge and each $\pi_i$ is a (possibly empty) subpath of $\pi$ that is an $(f/k)$-fault replacement path in $G$, with $q = O(k)$.

Our proof strategy will be similar to the previous argument with some minor changes: we still choose $\pi_i$ greedily as the longest subpath which is an $(f/k)$-fault replacement path, and we will take the next edge in the subpath as the $e_i$ to interweave.
Let $q$ be the number of subpaths resulting from this decomposition; our goal is to upper bound $q$ to be linear in $k$.

We will define $\pi_i'$ as $\pi_i$ augmented with $e_i$ (again $\pi'_q$ is undefined).
We define $x_i$ as the vertex at the end of $\pi_{i-1}'$ and at the beginning of $\pi_i$, so that for any $i$,
\[\pi(s,t \mid F)[x_i, x_{i+1}] = \pi_i' = \pi_i \circ e_i.\] 

Unlike in the unweighted setting, we no longer have overlaps in the $\pi_i'$. We will assess whether subpaths $\pi_i'$ are pre-light or post-light based on their weighted length, and proceed supposing that at least half of the subpaths are post-light. We generate FS-pairs with post-light subpaths as before, using the property that by maximality of $\pi_i$, each $\pi_i'$ necessarily fails to be an $(f/k)$-fault replacement path.
Using the same argument based on Hall's Theorem as before, this guarantees that we get at least $\frac{(q-1)f}{2k}$ distinct FS-pairs.

Now we can complete the proof of Theorem \ref{thm:intromainwtd} by the following lemma to limit the number of FS-pairs each fault is in, which is analogous to Lemma \ref{fsc} and has a similar proof. Theorem \ref{thm:intromainwtd} follows as we can again upper bound $\frac{(q-1)f}{2k}$ by $4f$ to bound $q$.
\begin{lemma}
    Each (weighted) fault in $F$ will be in at most 4 FS-pairs.
\end{lemma} 
\begin{proof}
Similarly to Lemma \ref{fsc} we will prove the lemma by showing that no fault can be in 5 FS-pairs. Suppose, for a contradiction, that we have fault $e = (u,v)$ in 5 FS-pairs. Without loss of generality at least 3 subpaths $\pi_a'$, $\pi_b'$, and $\pi_c'$ have fault-free paths which are contained in shortcuts from their right ends $x_{a+1}$, $x_{b+1}$, and $x_{c+1}$ to $u$. Let these paths be $p_a$, $p_b$, and $p_c$. Since each path is contained in a shortcut using $e$, we have 
\[w(p_a) < w(\pi_a') - w(e) \quad \text{and} \quad w(p_c) < w(\pi_c') - w(e).\]
Since $\pi_a'$ is post-light, we have 
\[w(\pi_{a+1}') \geq w(\pi_a').\]
Again we can use that $a+1 < c$ and that $\pi(s,t\mid F)$ is a shortest path to lower bound the weighted distance of $x_{a+1}$ to $x_{c+1}$ in  $G \setminus F$ as 
\begin{align*}
    d_{G \setminus F}(x_{a+1}, x_{c+1}) &= w(\pi(s, t \mid F)[x_{a+1}, x_{c+1}]) \\
    &= \sum_{i = a+1}^c w(\pi_i')\\
    &\geq w(\pi_{a+1}') + w(\pi_{c}') \\
    &\geq w(\pi_{a}') + w(\pi_{c}').
\end{align*}
However we can use the fault free paths of $p_a$ and $p_c$ to upper bound the distance from $x_{a+1}$ to $x_{c+1}$ in $G \setminus F$ to get a contradiction with the previous lower bound:
\begin{align*}
    d_{G \setminus F}(x_{a+1}, x_{c+1}) &\leq d_{G \setminus F}(x_{a+1}, u) + d_{G \setminus F}(u, x_{c+1})\\
    &\leq w(p_a) + w(p_c) \\
    &< w(\pi_a') + w(\pi_c') - 2 w(e). \tag*{\qedhere}
\end{align*}
\end{proof}
In the case that at least half of the subpaths are pre-light, we will generate FS-pairs with pre-light subpaths by defining base faults relative to the left ends $x_i$ of subpaths $\pi_i'$. In the analysis, we replace $x_{a+1}$, $x_{b+1}$, and $x_{c+1}$ with $x_a$, $x_b$ and $x_c$. Our analysis of $p_a$, $p_b$, and $p_c$ are unchanged. Comparing subpaths, we instead use the pre-light property of $\pi_c'$ to get
\[w(\pi_{c-1}') \geq w(\pi_{c}').\]
Then the analysis on the distance is a lower bound of
\begin{align*}
    d_{G \setminus F}(x_{a}, x_{c}) &= w(\pi(s, t \mid F)[x_{a}, x_{c}]) \\
    &= \sum_{i = a}^{c-1} w(\pi_i')\\
    &\geq w(\pi_{a}') + w(\pi_{c-1}') \\
    &\geq w(\pi_{a}') + w(\pi_{c}'),
\end{align*}
and an upper bound of
\begin{align*}
    d_{G \setminus F}(x_{a}, x_{c}) &\leq d_{G \setminus F}(x_{a}, u) + d_{G \setminus F}(u, x_{c})\\
    &\leq w(p_a) + w(p_c) \\
    &< w(\pi_a') + w(\pi_c') - 2 w(e). \tag*{\qedhere}
\end{align*}
\section{Algorithmic Path Decomposition \label{sec:algorithms}}

We will next prove Theorem \ref{thm:intromainalg}, which holds for unweighted input graphs, and then afterwards describe the (minor) changes needed to adapt the algorithm to the weighted setting.
As a reminder of our goal: we are given a graph $G$, a fault set $F$, a replacement path $\pi(s,t \mid F)$, and a parameter $k$ on input.
Our goal is to find nodes $\{x_i\}$ and fault sets $F_i$, which partitions $\pi(s,t \mid F)$ into $q=O(k)$ replacement paths avoiding $f/k$ faults each, as
\[\pi(s,t \mid F) = \pi(x_0,x_1 \mid F_1) \circ \pi(x_1, x_2 \mid F_2) \circ \ldots \circ \pi(x_{q-1}, x_q \mid F_q).\]

\subsection{Fault Set Reducing Subroutine}

Before describing our main algorithm, we will start with a useful subroutine, driven by an observation about the matching step in FS-pair generation. In our upper bound proof, we used a process for generating FS-pairs to bound the number of subpaths in the decomposition. We used \emph{minimum size} of the fault set $F_i$ associated to each augmented subpath $\pi_i'$ to argue that we could generate $|F_i|$ distinct FS-pairs.

The observation is that, letting $F_i$ be \emph{any} (not necessarily minimum) valid fault set for $\pi_i'$ (that is, $\pi_i'$ is a shortest path in $G \setminus F_i$), if we can produce an FS-pair for every fault in $F_i$ then our previous argument works.
On the other hand, if we cannot produce an FS-pair for every fault in $F_i$, then our previous argument gives us a process by which we can find a strictly smaller fault set $F_i'$ that is also valid for $\pi_i'$, by replacing the subset of $F_i$ with the reduced set of their base faults.  

The subroutine \textsc{FaultReduce} runs this process iteratively, in order to find a fault set $F_i$ for the input subpath $\pi_i$ that can be used to generate $|F_i|$ FS-pairs (from both the left and right).
We note the subtlety that $F_i$ is not necessarily a minimum valid fault set for $\pi_i$: as in Figure \ref{fig:minfaultset}, there may exist a smaller valid fault set, but the algorithm will halt nonetheless if it can certify that the appropriate number of FS-pairs can be generated.

\begin{figure} [h]
    \centering
    \includegraphics[scale = 1]{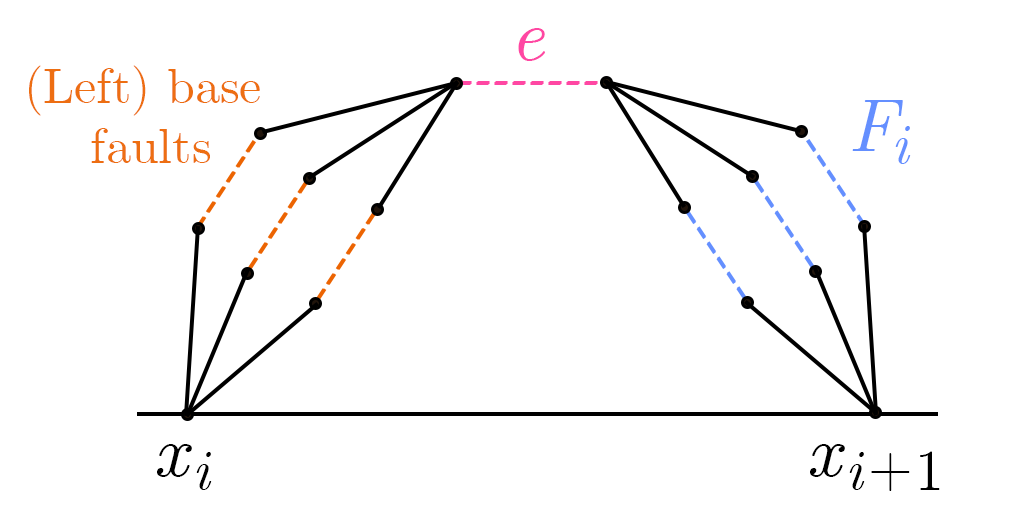}
    \caption{This subpath and fault set $F_i$ produces multiple FS-pairs via a saturated matching using the faults on the left as base faults, but its minimum fault set is only one edge $\{e\}$.}
    \label{fig:minfaultset}
\end{figure}

\begin{algorithm} [h]
\caption{\label{alg:faultreduce} \textsc{FaultReduce} $(\pi_i)$}
\begin{algorithmic}
    \State $F_i \leftarrow F$
    \State Construct $\Gamma_L$.
    \State Construct $\Gamma_R$.
    \Loop
        \State Compute max matching $M_L$ for $F_i$ in $\Gamma_L$.
        \State Compute max matching $M_R$ for $F_i$ in $\Gamma_R$.
        \If {$|M_L| < |F_i|$ or $|M_R| < |F_i|$}
            \State Reduce $F_i$
        \Else\space \Return {$F_i$}
        \EndIf
    \EndLoop
\end{algorithmic}
\end{algorithm}

The essential properties of Algorithm \ref{alg:faultreduce} are captured by the following lemma.

\begin{lemma} \label{lem:faultreducecorrect}
Relative to a graph $G$ and fault set $F$, there is a subroutine (Algorithm \ref{alg:faultreduce} - FaultReduce) that runs in polynomial time with the following behavior:
\begin{itemize}
\item The input is a path $\pi_i$ that is a shortest path in $G \setminus F$.

\item The output is a fault set $F_i \subseteq F$, such that:
\begin{itemize}
\item $\pi_i$ is a shortest path in $G \setminus F_i$, and

\item one can generate $|F_i|$ left- and $|F_i|$ right-FS-pairs of $\pi_i$ from $F_i$.
\end{itemize}
\end{itemize}
\end{lemma} 

We will next provide additional details on some of the steps in Algorithm \ref{alg:faultreduce}, and then prove Lemma \ref{lem:faultreducecorrect}.

\paragraph{Construction of $\Gamma_L$ and $\Gamma_R$.}
As in our previous proof, the graphs $\Gamma_L, \Gamma_R$ are the graphs representing the association between faults in $F_i$ and left or right (respectively) base faults in $F$.
More specifically:
\begin{itemize}
\item Both $\Gamma_L$ and $\Gamma_R$ are bipartite graphs with vertex set $F_i \cup F$, where $F_i$ is the current fault set, and $F$ is all initial faults.

Thus faults in $F_i$ are represented by two vertices, one on each side of the bipartition.

\item In $\Gamma_L$, we place an edge from $e \in F_i$ to $e_b \in F$ if and only if $e_b$ is a left base fault for $e$.
The edges of $\Gamma_R$ are defined similarly, with respect to right base faults.
\end{itemize}

These graph constructions require us to efficiently check whether or not a particular fault $e_b \in F$ acts as a (left or right) base fault for some $e \in F_i$. 
We next describe this process:
\begin{lemma}
Given a subpath $\pi_i$, a valid fault set $F_i$, and faults $e \in F_i, e_b \in F$, we can check whether or not $e_b$ is a left and/or right base fault of $e$ in polynomial time.
\end{lemma}
\begin{proof}
First, the following notation will be helpful.
Let $x_i, x_{i+1}$ be the endpoints of the input subpath $\pi_i$.
We will write
$d(x_i, x_{i+1} \mid e_b \leadsto e)$
for the length of the shortest (possibly non-simple) $(x_i, x_{i+1})$-path that contains both $e_b$ and $e$, and which specifically uses $e_b$ as the first fault in $F$ along the path. 
We define $d(x_{i+1}, x_{i} \mid e_b \leadsto e)$ similarly.
Note that $e_b$ is a left base fault for $e$ if and only if
$$d(x_i, x_{i+1} \mid e_b \leadsto e) < |\pi_i|$$
and that $e_b$ is a right base fault for $e$ if and only if
$$d(x_{i+1}, x_{i} \mid e_b \leadsto e) < |\pi_i|.$$
Thus, it suffices to compute the values of the left-hand side of these two inequalities.
We will next describe computation of $d(x_i, x_{i+1} \mid e_b \leadsto e)$; the other computation is symmetric.
There are two cases, depending on whether or not $e_b = e$. Let $e = (u,v)$, $e_b = (u_b,v_b)$.
When $e_b \ne e$, the formula is:
\begin{align*}
      d(x_i, x_{i+1} \mid e_b \leadsto e)  = \min\{&d_{G \setminus F}(x_i, u_b) + d_{G}(v_b, u) + d_G(v, x_{i+1}) + 2, \\
     &d_{G \setminus F}(x_i, u_b) + d_{G}(v_b, v) + d_{G}(u, x_{i+1}) + 2, \\
     &d_{G \setminus F}(x_i, v_b) + d_{G}(u_b, u) + d_{G}(v, x_{i+1}) + 2, \\
     &d_{G \setminus F}(x_i, v_b) + d_{G}(u_b, v) + d_{G}(u, x_{i+1}) + 2
     \}.
\end{align*}
The four parts are needed since we consider paths that use $e, e_b$ with either orientation, and the $+2$ term arises to count the contribution of the edges $e, e_b$ themselves.
In the case where $e_b = e$, the formula is
\begin{align*}
    d(x_i, x_{i+1} \mid e \leadsto e) &= \min\{d_{G\setminus F}(x_i, u) + d_{G}(v, x_{i+1})+ 1, d_{G\setminus F}(x_i, v) + d_{G}(u, x_{i+1})+ 1\}. \tag*{\qedhere}
\end{align*}
\end{proof}

\paragraph{Reducing $F_i$.}

Next, we provide more detail on the step of reducing the fault set $F_i$.
This uses Hall's condition, in an analogous way to our previous proof.
When we compute max matchings $M_L, M_R$ for $\Gamma_L, \Gamma_R$, if we successfully find matchings of size $|M_L| \ge |F_i|$ or $|M_R| \ge |F_i|$, then we have certified the ability to generate $|F_i|$ left and right FS-pairs as in Section \ref{sec:fspairs}, and so the algorithm can return $F_i$ and halt.
Otherwise, suppose without loss of generality that $|M_L| < |F_i|$.
By Hall's condition, that means there exists a fault subset $A \subseteq F_i$ such that the set of base faults $B \subseteq F$ used by faults in $A$ is strictly smaller than $A$ itself.
For the reduction step, we set $F_i \gets F_i \cup B \setminus A$, which reduces the size of $|F_i|$.
By Lemma \ref{lem:hallreduce2}, this maintains the invariant that $F_i$ is a valid fault set for the input path $\pi_i$.

In order to efficiently find the non-expanding fault subset $A \subseteq F_i$, we may compute the max matching in $\Gamma_L$ (or $\Gamma_R$) using a primal-dual algorithm that returns both a max matching and a certificate of maximality of this form.
For example, the Hungarian algorithm will do \cite{CLRS}.

\subsection{Main Algorithm}

\textsc{ComputeSubpaths}, described in Algorithm \ref{alg:computesubpaths}, performs a greedy search for subpath boundaries.
In each round, we set the next subpath boundary node $x_{i+1}$ to be the furthest node from the previous subpath boundary node $x_i$, such that the corresponding subpath is certified by the algorithm \textsc{FaultReduce} to have size at most $f/k$.
Thus, considering the augmented subpath that we get by adding an additional node to $\pi_i$, we can generate at more than $f/k$ left and right FS-pairs from this subpath.

We next state the algorithm; for ease of notation we label the vertices of the input path $\pi(s,t \mid F)$ as $v_0, v_1, \ldots, v_\ell$.

\begin{algorithm}[h]
\caption{\textsc{ComputeSubpaths} $(\pi(s,t \mid F), F, k)$}\label{alg:computesubpaths}
\begin{algorithmic}
    \State $x_0 \leftarrow s$.
    \State $i \leftarrow 0$.
    \While {$x_i \neq v_\ell$}
        \State Binary search for largest $y$ such that the fault set returned by \textsc{FaultReduce}$(\pi(s,t)[x_i, v_y])$ has size $\leq f/k$.
        \State $i \leftarrow i+1$.
        \State $x_i \leftarrow v_y$.
    \EndWhile\\
    \Return $\{x_j\}_{j=0}^i$
\end{algorithmic}
\end{algorithm}
\begin{theorem}
Algorithm 2 is correct and runs in polynomial time. 
\end{theorem}
\begin{proof}
    In Corollary \ref{cor:fsbound} from our upper bound section, we showed that there exist only $O(f)$ total right FS-pairs using post-light subpaths (and, symmetrically, there exist only $O(f)$ left FS-pairs using pre-light subpaths). Since at least half of the augmented subpaths are pre-light or half are post-light, and by Lemma \ref{lem:faultreducecorrect} every augmented subpath can generate at least $f/k$ left and right FS-pairs, altogether we will have at most $O(k)$ subpaths.
    
    For runtime, we always generate a linear number of subpaths, and locating the endpoint of each requires calling the subroutine $\log n$ times. Thus the entire algorithm runs in polynomial time. 
\end{proof}

A similar approach works in the weighted setting, since the method of counting FS-pairs extends to the structure in Theorem \ref{thm:intromainwtd} and upper bounds the number of interweaved subpaths and edges. The construction of auxiliary graphs $\Gamma_L$ and $\Gamma_R$ requires checking the weighted distance, but the matching and FS-pair generation is the same. We change the algorithm to add the next edge into the decomposition of $\pi(s,t \mid F)$ after finding a maximal subpath with fault set at most $f/k$. The upper bound for the number of subpaths based on enough FS-pairs being generated follows from the analysis in Theorem \ref{thm:intromainwtd}.

\section*{Acknowledgments}

We are grateful to Vijaya Ramachandran for helpful references to prior work, and to an anonymous reviewer for exceptionally helpful and thorough comments.

\bibliographystyle{alpha}
\bibliography{main}

\newcommand{\etalchar}[1]{$^{#1}$}
\begin{thebibliography}{ABBK{\etalchar{+}}02}

\bibitem[ABBK{\etalchar{+}}02]{restoration}
Yehuda Afek, Anat Bremler-Barr, Haim Kaplan, Edith Cohen, and Michael Merritt.
\newblock Restoration by path concatenation: fast recovery of mpls paths.
\newblock {\em Distributed Computing}, 15:273--283, 2002.

\bibitem[ABDS{\etalchar{+}}20]{spanners}
Reyan Ahmed, Greg Bodwin, Faryad Darabi~Sahneh, Stephen Kobourov, and Richard Spence.
\newblock Weighted additive spanners.
\newblock In Isolde Adler and Haiko M{\"u}ller, editors, {\em Graph-Theoretic Concepts in Computer Science}, pages 401--413, Cham, 2020. Springer International Publishing.

\bibitem[BCG{\etalchar{+}}18]{BCGLPP18}
Davide Bil{\`o}, Keerti Choudhary, Luciano Gual{\`a}, Stefano Leucci, Merav Parter, and Guido Proietti.
\newblock Efficient oracles and routing schemes for replacement paths.
\newblock In {\em 35th Symposium on Theoretical Aspects of Computer Science (STACS 2018)}. Schloss Dagstuhl-Leibniz-Zentrum fuer Informatik, 2018.

\bibitem[BGPW17]{BGPV17}
Greg Bodwin, Fabrizio Grandoni, Merav Parter, and Virginia~Vassilevska Williams.
\newblock {Preserving Distances in Very Faulty Graphs}.
\newblock In Ioannis Chatzigiannakis, Piotr Indyk, Fabian Kuhn, and Anca Muscholl, editors, {\em 44th International Colloquium on Automata, Languages, and Programming (ICALP 2017)}, volume~80 of {\em Leibniz International Proceedings in Informatics (LIPIcs)}, pages 73:1--73:14, Dagstuhl, Germany, 2017. Schloss Dagstuhl--Leibniz-Zentrum fuer Informatik.

\bibitem[BP21]{tiebreaking}
Greg Bodwin and Merav Parter.
\newblock Restorable shortest path tiebreaking for edge-faulty graphs.
\newblock In {\em Proceedings of the 2021 ACM Symposium on Principles of Distributed Computing}, PODC'21, page 435–443, New York, NY, USA, 2021. Association for Computing Machinery.

\bibitem[CC19]{CC19}
Shiri Chechik and Sarel Cohen.
\newblock Near optimal algorithms for the single source replacement paths problem.
\newblock In {\em Proceedings of the Thirtieth Annual ACM-SIAM Symposium on Discrete Algorithms}, pages 2090--2109. SIAM, 2019.

\bibitem[CCFK17]{CCFK17}
Shiri Chechik, Sarel Cohen, Amos Fiat, and Haim Kaplan.
\newblock (1+$\varepsilon$)-approximate f-sensitive distance oracles.
\newblock In {\em Proceedings of the Twenty-Eighth Annual ACM-SIAM Symposium on Discrete Algorithms}, pages 1479--1496. SIAM, 2017.

\bibitem[CLRS09]{CLRS}
Thomas~H. Cormen, Charles~E. Leiserson, Ronald~L. Rivest, and Clifford Stein.
\newblock {\em Introduction to Algorithms, 3rd Edition}.
\newblock {MIT} Press, 2009.

\bibitem[GJM20]{GJM20}
Manoj Gupta, Rahul Jain, and Nitiksha Modi.
\newblock Multiple source replacement path problem.
\newblock In {\em Proceedings of the 39th Symposium on Principles of Distributed Computing}, pages 339--348, 2020.

\bibitem[HS01]{HS01}
John Hershberger and Subhash Suri.
\newblock Vickrey prices and shortest paths: What is an edge worth?
\newblock In {\em Proceedings 42nd IEEE symposium on foundations of computer science}, pages 252--259. IEEE, 2001.

\bibitem[KIM82]{KIM82}
Naoki Katoh, Toshihide Ibaraki, and Hisashi Mine.
\newblock An efficient algorithm for k shortest simple paths.
\newblock {\em Networks}, 12(4):411--427, 1982.

\bibitem[LYXsS20]{liu2020survivability}
Bao~Ju Liu, Peng Yu, Qiu Xue-song, and Lei Shi.
\newblock Survivability-aware routing restoration mechanism for smart grid communication network in large-scale failures.
\newblock {\em EURASIP journal on wireless communications and networking}, 2020:1--21, 2020.

\bibitem[MMG89]{MMG89}
Kavindra Malik, Ashok~K Mittal, and Santosh~K Gupta.
\newblock The k most vital arcs in the shortest path problem.
\newblock {\em Operations Research Letters}, 8(4):223--227, 1989.

\bibitem[MR22]{MR22}
Vignesh Manoharan and Vijaya Ramachandran.
\newblock Near optimal bounds for replacement paths and related problems in the congest model.
\newblock {\em arXiv preprint arXiv:2205.14797}, 2022.

\bibitem[WGY{\etalchar{+}}15]{WGYJ+15}
Rui Wang, Suixiang Gao, Wenguo Yang, Zhipeng Jiang, et~al.
\newblock Restorable energy aware routing with backup sharing in software defined networks.
\newblock {\em J. Commun.}, 10(8):551--561, 2015.

\bibitem[ZGK{\etalchar{+}}21]{ZGKLXZ21}
Zhizhen Zhong, Manya Ghobadi, Alaa Khaddaj, Jonathan Leach, Yiting Xia, and Ying Zhang.
\newblock Arrow: restoration-aware traffic engineering.
\newblock In {\em Proceedings of the 2021 ACM SIGCOMM 2021 Conference}, pages 560--579, 2021.

\end{thebibliography}

\end{document}